%% file: synchrony_journal.tex
\documentclass{article}
\usepackage{arxiv}

\input{preamble_arXiv.tex}

\input{notation_defs.tex}

\usepackage{cite}
\usepackage{amsmath,amssymb,amsfonts}
\usepackage{algorithmic}
\usepackage{graphicx}
\usepackage{algorithm,algorithmic}
\usepackage{textcomp}

\usepackage{lmodern}

\newcommand{\dds}{\frac{\td}{\td s}}
\newcommand{\at}[2]{\left. #1 \right| _{#2}}
\newcommand{\ddso}{\at{\dds}{s=0}}

\newcommand{\Lie}{\mathrm{Lie}}
\newcommand{\Flow}{\mathfrak{F}}

\newcommand{\Diff}{\mathfrak{D}}

\newcommand{\set}[2]{\left\{ #1 \;\middle|\; #2 \right\}}

\newcommand{\publicationdetails}
{Under review

This work was supported by the European Union through the Horizon Europe research and innovation program [grant agreement No. ID: 101154194] MEW, and by the Australian Research Council through the Discovery Project under Grant DP210102607.}
\newcommand{\publicationversion}
{Preprint}

\begin{document}

\title{Synchronous Models and Fundamental Systems in Observer Design}
\headertitle{Synchronous Models and Fundamental Systems in Observer Design}

\author{
\href{https://orcid.org/0000-0003-4391-7014}{\includegraphics[scale=0.06]{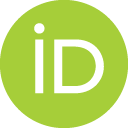}\hspace{1mm}
Pieter van Goor}
\\
    Robotics and Mechatronics group \\
    University of Twente \\
    7500 AE Enschede, The Netherlands \\
    \texttt{Pieter.vanGoor@anu.edu.au} \\
	\And	\href{https://orcid.org/0000-0002-7803-2868}{\includegraphics[scale=0.06]{orcid.png}\hspace{1mm}
    Robert Mahony}
\\
    Systems Theory and Robotics Group \\
	Australian National University \\
    ACT, 2601, Australia \\
	\texttt{Robert.Mahony@anu.edu.au} \\
}

\maketitle

\vspace{1cm}



\begin{abstract}
This paper introduces the concept of a \emph{synchronous model} as an extension of the internal model concept used in observer design for dynamical systems.
A system is said to contain a synchronous model of another if there is a suitable error function between the two systems that remains stationary for all of the trajectories of the two systems. 
A system is said to admit a synchronous lift if a second system containing a synchronous model exists.
We provide necessary and sufficient conditions that a system admits a synchronous lift and provide a method to construct a (there may be many) lifted system should one exist. 
We characterise the class of all systems that admit a synchronous lift by showing that they consist of fundamental vector fields induced by a Lie group action, a class of system we term \emph{fundamental systems}. 
For fundamental systems we propose a simple synchronous observer design methodology, for which we show how correction terms can be discretised and combined easily, facilitating global characterisation of convergence and performance.
Finally, we provide three examples to demonstrate the key concepts of synchrony, symmetry construction, and observer design for a fundamental system.
\end{abstract}


\section{Introduction}

The importance of Lie groups in estimation of nonlinear systems has been studied since at least the 1970s \cite{1972_brockett_SystemTheoryGroup,1985_grizzle_StructureNonlinearControl,1990_cheng_ObservabilitySystemsLie,1991_salcudean_GloballyConvergentAngular}.
Equivariance- and invariance-based observers and filters have seen increasing impact in the state estimation literature over the past decades, particularly for applications in robotics, navigation, and aerospace 
\cite{1985_grizzle_StructureNonlinearControl,2007_bonnabel_ObservateursAsymptotiquesInvariants,2018_barrau_InvariantKalmanFiltering,2022_mahony_ObserverDesignNonlinear,2023_vangoor_EquivariantFilterEqF}.
In the 2000s, the advent of aerial robotics renewed interest in strategies for attitude estimation on the special orthogonal group $\SO(3)$ \cite{2003_markley_AttitudeErrorRepresentations,2007_bonnabel_LeftinvariantExtendedKalman,2008_mahony_NonlinearComplementaryFilters}, as well as estimation of only a single direction on the sphere $\Sph^2$ treated as a homogeneous space of $\SO(3)$ \cite{2006_hamel_AttitudeEstimation3,2012_grip_AttitudeEstimationUsing,2012_trumpf_AnalysisNonLinearAttitude}.
A common feature of those observer designs that exhibited (almost-)global stability was that the uncorrected error dynamics remained stationary, allowing Lyapunov analysis to be undertaken without the need to compensate or dominate exogenous terms introduced by the input signal.
This property was formally defined as \emph{synchrony} in \cite{2010_lageman_GradientLikeObserversInvariant}, based on earlier similar notions such as in \cite{1997_nijmeijer_ObserverLooksSynchronization}.
Synchrony was shown to facilitate powerful observer designs for left- and right-invariant systems on Lie groups in \cite{2010_lageman_GradientLikeObserversInvariant,2013_khosravian_BiasEstimationInvariant,2015_khosravian_ObserversInvariantSystems}.
This concept of synchrony is not the same as synchronisation in coupled and multi-agent systems \cite{1997_nijmeijer_ObserverLooksSynchronization,1998_pogromsky_PassivityBasedDesign,2007_pham_StableConcurrentSynchronization,2007_stan_AnalysisInterconnectedOscillators}, that refers to the convergence of individual system trajectories to a common solution of the individual dynamics. 
Rather, the synchrony we use is the property that there is an error or measure between two systems that is constant along all trajectories of the systems when driven by the same input signal.  
The power of synchrony in observer design is that if the observer and true systems are synchronous in this sense, then the associated error dynamics depend only on the correction term added to the observer. 
This removes the need for the correction terms to compensate or dominate system dynamics, and is essential to many of the observer designs that feature (almost-)global (rather than local or semi-global) asymptotic stability of the error dynamics.

\emph{Group-affine systems} on Lie groups were first described in the observer design literature in \cite{2014_barrau_InvariantParticleFiltering}, although they had been studied earlier in the context of controllability and reachability under the name of \emph{Group-linear systems} \cite{1999_ayala_LinearControlSystems,2001_ayala_ControllabilityPropertiesClass,2010_jouan_EquivalenceControlSystems,2013_ayala_NullControllabilityLie}.
The importance of group-affine systems in estimation is due largely to their log-linearity property, which has been exploited in linearisation-based stochastic observers, such as the invariant extended Kalman filter (IEKF) \cite{2017_barrau_InvariantExtendedKalman,2018_barrau_InvariantKalmanFiltering,2019_brossard_ExploitingSymmetriesDesign,2022_brossard_AssociatingUncertaintyExtended} and the Equivariant Filter (EqF) \cite{2022_mahony_ObserverDesignNonlinear,2023_vangoor_EquivariantFilterEqF,2023_fornasier_EquivariantSymmetriesInertial}.
While the earlier observer designs for left-invariant systems based on \cite{2010_lageman_GradientLikeObserversInvariant} do not exhibit synchrony for group-affine systems, it was recently shown in \cite{2021_vangoor_AutonomousErrorConstructive} that synchrony could be recovered by designing the observer using a larger observer space that included the automorphism group of the original Lie group.
This design approach has since been exploited to yield the first almost-global observer designs for problems in inertial navigation \cite{2023_vangoor_ConstructiveEquivariantObserver,2025_vangoor_SynchronousObserverDesign}.
A number of studies have sought to characterise systems that admit invariant or group-affine dynamics \cite{2010_lageman_GradientLikeObserversInvariant,2010_jouan_EquivalenceControlSystems,2021_vangoor_AutonomousErrorConstructive,2024_liu_ExistenceLinearObserved}.
However, these studies do not directly consider the property of synchrony and as a consequence do not fully resolve the underlying system structure. 
To the authors' knowledge, the present work provides the first intrinsic characterisation of synchrony and its deep ties to Lie group symmetries leading to a new class of \emph{fundamental systems}. 

In this paper we formally introduce the concept of a system containing a \emph{synchronous model} of another system.
This concept is an extension of the idea of a system containing an internal model of another system and builds on earlier notions of synchrony \cite{1990_nijmeijer_NonlinearDynamicalControl,2010_lageman_GradientLikeObserversInvariant}.
Systems containing a synchronous model of a target system are ideal candidates, indeed a better choice than the target system dynamics themselves, for the observer dynamics of an observer architecture.
However, existence of a synchronous model for system is more demanding than existence of an internal model, and in fact, we show that the existence of a synchronous model is only possible if the original system has specific differential geometric structure.
Specifically, we show that a synchronous model can be constructed if and only if the accessibility Lie algebra of the original system is finite-dimensional and complete.
This class of systems is closely related to \emph{fundamental systems}, which are those induced by the fundamental vector fields associated with a group action on a homogeneous space, and include the well-studied left-, right-, and dual-, invariant systems on Lie groups.
The main result of the paper (Theorem~\ref{thm:synchrony_fundamental}) states that a controllable system admits a synchronous lift if and only if it is fundamental. 
Finally, we propose a hybrid observer design for a fundamental system that exploits synchrony in developing and combining correction terms.

The structure of the paper is as follows, aside from the introduction and conclusion.
Section \ref{sec:preliminaries} introduces the prerequisite differential geometry, Lie theory, and systems theory.
Section \ref{sec:synchronous_model} formally defines the notions of an error function and a synchronous model as an extension of the concept of internal model.
Section \ref{sec:symmetry_construction} provides the key technical results that characterise how the existence of a synchronous model for a system is intrinsically linked to that system's accessibility Lie algebra.
Section \ref{sec:fundamental_systems} defines fundamental systems and shows that a controllable system admits a synchronous model if and only if it is fundamental.
Section \ref{sec:observer_Design} contains two theorems that support synchronous observer designs for fundamental systems, showing how one can easily discretise and combine correction terms by exploiting synchrony.
Section \ref{sec:observer_Design} details three examples that demonstrate a fundamental system, the construction of a synchronous model, and the application of the discretisation and combination theorems.

\section{Preliminaries}
\label{sec:preliminaries}

For a comprehensive introduction to smooth manifolds and Lie groups, we refer the reader to \cite{2012_lee_IntroductionSmoothManifolds}.

\subsection{Smooth Manifolds}
Let $\calM$ be a smooth $m$ dimensional manifold.
The tangent space of $\calM$ at a point $\xi \in \calM$ is denoted $\tT_\xi \calM$, and the tangent bundle is denoted $\tT \calM$.
For a smooth map $h : \calM \to \calN$, the differential of $h$ at a point $\xi \in \calM$ is a linear map denoted
\begin{align*}
    \tD h(\xi) : \tT_\xi \calM \to \tT_{h(\xi)}\calN.
\end{align*}
To disambiguate the manifold and tangent space arguments we write the vector argument $v \in \tT_\xi \calM$ in square brackets $\tD h(\xi)[v] \in \tT_{h(\xi)}\calN$.
The differential of $h$, as a bundle homomorphism of tangent bundles, is denoted
\begin{align*}
    \tD h : \tT \calM \to \tT \calN,
\end{align*}
and is used when the base point of the argument is clear from context.
For a smooth map with multiple arguments, $h : \calM_1 \times \calM_2 \to \calN$, the differential with respect to a particular argument $\xi_2$ at a particular value $\zeta_2$ is denoted
\begin{align*}
    \tD_{\xi_2 | \zeta_2} h(\xi_1,\xi_2) : \tT_{\zeta_2} \calM_2 \to \tT_{h(\xi_1, \zeta_2)}\calN,
\end{align*}
and the other arguments are taken to be fixed.

A diffeomorphism $\varphi : \calM \to \calN$ is a smooth map between smooth manifolds $\calM$ and $\calN$ with a smooth inverse.
The space of diffeomorphisms from $\calM$ to itself is labeled $\Diff(\calM)$.
The space of smooth vector fields on $\calM$ is denoted $\gothX(\calM)$.
For a given vector field $f \in \gothX(\calM)$, $\Flow_f^t$ is the (maximal) flow of $f$ for a time $t$, and satisfies
\begin{align}
    \ddt \Flow_f^t(\xi) &= f(\Flow_f^t(\xi)), &
    \Flow_f^0(\xi) &= \xi,
    \label{eq:vector_field_flow}
\end{align}
for all $\xi \in \calM$.
If the flow $\Flow_f^t$ exists for all time $t \in \R$ and all $\xi \in \calM$, then the vector field $f$ is said to be complete.
In this case, $\Flow_f^t: \calM \to \calM$ is a diffeomorphism, $\Flow_f^t \in \Diff(\calM)$.

The vector fields $\gothX(\calM)$ form a Lie algebra under the classical Lie bracket $[f,g] \in \gothX(\calM)$ defined for any $f, g \in \gothX(\calM)$ by
\begin{align}\label{eq:lie-bracket-dfn}
    [f,g](\xi) :=
\at{\ddt}{t=0} \Flow_{g}^{-\sqrt{t}}\circ\Flow_{f}^{-\sqrt{t}}\circ\Flow_{g}^{\sqrt{t}}\circ\Flow_{f}^{\sqrt{t}} (\xi).
\end{align}
We also use the notation $\ad_{f}g = [f,g]$ particularly where a differential acts on a Lie bracket of vector fields (e.g.~Lemma~\ref{lem:synchronous_vector_fields}).
In any given set of local coordinates the Lie bracket can also be computed by
\begin{align*}
    [f,g](\xi) = \tD g(\xi) [f(\xi)] - \tD f(\xi) [g(\xi)].
\end{align*}

\subsection{Systems}
Let $\vecL$ be a finite dimensional (real) vector space and let $f : \vecL \to \gothX(\calM)$, $f(v) := f_v \in \gothX(\calM)$, be a map that we term the \emph{system map}.
The associated system is the set of trajectories $(\xi(t; \xi_0), v(t)) \in \calM \times \vecL$, for $t \in [0, \infty)$, such that
\begin{align}
\dot{\xi} = f_{v(t)} (\xi(t)),  \qquad\qquad \xi_0 \in \calM.
\label{eq:system}
\end{align}
In practice, the trajectories are restricted to solutions associated with input signals $v(t) \in \vecL$ that satisfy real-world constraints, typically piecewise uniformly Lipschitz continuous in time.
The exact assumptions here are not important as long as the signal class considered ensures local existence and uniqueness of solutions $\xi(t)$.
Global existence of solutions is related to completeness of the vector fields $f_v$ and is separately discussed in the results presented later in the paper.
We will assume that $f : \vecL \to \gothX(\calM)$ is affine.
That is,
\begin{align}
f_v = f_0 + \sum_i^\ell v_i f_i,
\label{eq:affine_system}
\end{align}
where $\{\eb_1, \ldots, \eb_\ell \}$ is a basis for $\vecL$ and $v = \sum_i^\ell v_i \eb_i$.
Here $f_0 \in \gothX(\calM)$ is termed the \emph{drift} vector field and the $f_i := f_{\eb_i} - f_0 \in \gothX(\calM)$ are termed the \emph{input} vector fields.

Given a system $f : \vecL \to \gothX(\calM)$, the image of the system map $\image(f) \subset \gothX(\calM)$ is an affine subspace of $\gothX(\calM)$.
The Lie algebra generated by the system $\Lie(\image (f))$, also called the \emph{accessibility algebra}, is defined as the smallest Lie sub-algebra of $\gothX(\calM)$ that contains $\image (f)$ \cite[Def.~3.7]{1990_nijmeijer_NonlinearDynamicalControl}.
A useful characterisation is given by \cite[Prop.~3.8]{1990_nijmeijer_NonlinearDynamicalControl};
\begin{align}
\Lie(\image (f)) &=  \mathrm{span}
 \cset{ [f_{{k_n}}, [\cdots,[f_{k_2},f_{k_1}] \cdot\cdot] \in \gothX(\calM)}
{n \in \N, \; k_1,...,k_n \in \{0, \ldots, \ell\} }.
\label{eq:system_lie_algebra}
\end{align}
Note that $\Lie(\image(f))$ may be infinite dimensional, although it is always countable by the construction \eqref{eq:system_lie_algebra}.
The accessibility distribution is the smooth assignment of subspaces 
\begin{align}
D_{\Lie(\image (f))}(\xi) = \spn \cset{ g(\xi) \in \tT_\xi \calM }{g \in \Lie(\image (f))}
\label{eq:accessibility_distribution}
\end{align} 
of the tangent bundle $\tT \calM$. 
A system is said to be \emph{controllable} if its accessibility distribution $D_{\Lie(\image (f))}(\xi) = \tT_\xi \calM$ spans the tangent space of $\calM$ at each point.

\subsection{Lie groups and algebras}

A Lie group $\grpG$ is a smooth manifold with smooth product, identity, and smooth inverse.
The product of elements $X,Y \in \grpG$ is denoted $X Y = \tL_X Y = \tR_Y X \in \grpG$ where $\tL$ and $\tR$ are the left and right translations.
The identity element is denoted $I\in \grpG$, and the inverse of an element $X \in \grpG$ is written $X^{-1} \in \grpG$.

The Lie algebra $\gothg$ of $\grpG$ is a vector space equipped with a Lie bracket $[\cdot, \cdot]: \gothg \times \gothg \to \gothg$, and may be identified with the tangent space of $\grpG$ at the identity $I$.
For any tangent vector $V \in \tT_{X} \grpG$, we denote the left- and right-translation of $V$ by $Y$ as $Y V = \tD \tL_Y [V] \in \tT_{Y X} \grpG$ and as $V Y = \tD \tR_Y [V]  \in \tT_{X Y} \grpG$, respectively.
In particular, for an element $U \in \gothg$ we use the notations $XU, UX \in \tT_X \grpG$.
Define the vector fields $\tR^\sharp_U(X) := XU \in \gothX(\grpG)$ (see \eqref{eq:fundamental_vector_field} below), then the Lie bracket of vector fields on $\gothX(\grpG)$ corresponds with the Lie bracket of $\gothg$,
\[
[\tR^\sharp_U, \tR^\sharp_V] = \tR^\sharp_{[U,V]}.
\]
Each vector field $\tR^\sharp_U$ is left-invariant in the sense that $Z \tR^\sharp_U(X) = Z X U = \tR^\sharp_U(ZX)$, for all $X,Z \in \grpG$.
A left invariant system on a group $\grpG$ is a dynamical system
\begin{align}
    \ddt X(t) &= \tR^\sharp_{U(t)}(X(t)) = X(t) U(t),
\label{eq:XU}
\end{align}
for $U(t) \in \gothg$ an input signal.

A (right) group action $\phi : \grpG \times \calM \to \calM$ of a Lie group $\grpG$ acting on a manifold $\calM$ is a smooth map that satisfies
\begin{align*}
    \phi(Y, \phi(X, \xi)) = \phi(XY, \xi), && \phi(I,\xi) = \xi.
\end{align*}
An action is transitive if the partial map $\phi_\xi : \grpG \to \calM$ given by $\phi_X(\xi) := \phi(X,\xi)$ is a surjective submersion.
Every $U \in \gothg$ induces a \emph{fundamental vector field}\footnote{
In the geometric mechanics literature, fundamental vector fields are often written simply $u^\sharp$ where the Lie group and group action are implicit in the problem formulation.
We introduce the notation $\phi^\sharp_u$ to make the group action explicit. 
}
$\phi^\sharp_U \in \gothX(\calM)$ defined by
\begin{align}
\phi^\sharp_U(\xi) := \tD \phi_\xi(I)[U].
\label{eq:fundamental_vector_field}
\end{align}
Since the map $\phi$ induces a smooth group homomorphism\footnote{If we consider $\phi$ as a left (rather than right) action, then $\phi : \grpG \to \Diff(\calM)$ is a Lie group anti-homomorphism.} 
$\phi : \grpG \to \Diff(\calM)$, $X \mapsto \phi_X$, from a Lie group $\grpG$ to diffeomorphisms of $\calM$, it induces a Lie algebra homomorphism $\phi^\sharp : \gothg \to \gothX(\calM)$.
Indeed, left-invariant vector fields on a Lie group are exactly the fundamental vector fields $\tR^\sharp_{U}(X) = \tD_{Z|\Id} \tR(Z,X)[U] = X U$. 

The exponential map $\exp : \gothg \to \grpG$ is defined by $\exp(t U) := \Flow_{\tR^\sharp_U}^t(I)$ for each $U \in \gothg$ and $t \in \R$.
For any group action $\phi$, the exponential also provides the flow of $\phi^\sharp_U$ by $\phi_{\exp(tU)}(\xi) = \Flow_{\phi^\sharp_U}^t(\xi)$.

\section{Synchronous Models}
\label{sec:synchronous_model}

In this section we recall the concept of an internal model before developing the notion of an error function and a synchronous model.
We emphasise that these concepts are asymmetric, in the sense that a system containing an internal or a synchronous model of another does not imply the converse.
An internal model can be defined directly in terms of system trajectories and a reconstruction map that converts the state space of one system into that of the other, while the definition of a synchronous model requires an error function that specifies how trajectories of the systems are to be compared.

Consider a system \eqref{eq:system} with system function $f : \vecL \to \gothX(\calM)$, where $\vecL$ is a finite-dimensional vector space and  $\calM$ is a smooth manifold.
Now consider a second system with system function $f^\dag : \vecL \to \gothX(\calG)$, where $\calG$ is another smooth manifold.
The system $f^\dag$ is said to \emph{contain an internal model} \cite{2010_lageman_GradientLikeObserversInvariant} of $f$ if there exists a surjective submersion $\varphi : \calG \to \calM$ (the reconstruction map) such that $\xi(t) = \varphi (\hat{X}(t))$ is a solution of $f$ for all solutions $\hat{X}(t)$ of $f^\dag$.
Since for any initial condition $\xi_0 \in \calM$ there is an initial condition $\hat{X}_0 \in \calG$ such that $\varphi(\hat{X}_0) = \xi_0$, then uniqueness of solutions implies that $f^\dag$ can reproduce all trajectories of $f$.
The property that $f^\dag$ contains an internal model is non-reflexive; that is, it says nothing about the ability of $f$ to reproduce trajectories of $f^\dag$.
Indeed, the requirement that the reconstruction map is surjective implies that $\dim \calG \geq \dim \calM$ and thus $f^\dag$ can have dynamics in the kernel of $\varphi$ that have nothing to do with the underlying system $f$.

Existence of an internal model says nothing about how the trajectories of $f$ and $f^\dag$ relate if the initial conditions do not correspond; $\varphi(\hat{X}_0) \not= \xi_0$.
A synchronous model is an extension of the concept of internal model that allows one to compare trajectories on $\calG$ and $\calM$ with different initial conditions. 
A synchronous model will depend on the notion of an \emph{error function}.

\begin{definition}[Error Function]\label{dfn:error_function}
    Let $\calM$ and $\calG$ be smooth manifolds.
    An \emph{error function} (of $\calM$ with respect to $\calG$) is a smooth map $e : \calG \times \calM \to \calM$, where:
    \begin{enumerate}
        \item\label{item:error_embedding} The partial maps $e_{\hat{X}} : \calM \to \calM$, $e_{\hat{X}}(\xi) := e(\hat{X}, \xi)$ are a family of diffeomorphisms of $\calM$.
        \item\label{item:error_reconstruction} For any fixed point $\mr{\xi} \in \calM$, the map $\phi_{\mr{\xi}} : \calG \to \calM$,  $\phi_{\mr{\xi}}(\hat{X}) := e_{\hat{X}}^{-1}(\mr{\xi})$ is a smooth surjective submersion.
    \end{enumerate}
\hfill$\Box$\end{definition}

This definition of error is, importantly, asymmetric, analogous to the definition of an internal model.
That is, while the error function is a diffeomorphism from $\calM$ to itself for every fixed $\hat{X} \in \calG$, it is not the case that the error must embed $\calG$ into $\calM$ for every (or even any) fixed $\xi \in \calM$.
This will allow flexibility in choosing $f^\dag$ dynamics while ensuring that all the $f$ dynamics are captured when the definition is used in observer design.
Given the error map $e$, then Condition \ref{item:error_reconstruction} means that any choice of fixed point $\mr{\xi} \in \calM$ induces a reconstruction map,
\begin{align}\label{eq:reconstruction_map}
    \varphi(\hat{X}) := \phi_{\mr{\xi}}(\hat{X}) = e_{\hat{X}}^{-1}(\mr{\xi}).
\end{align}
This construction motivates why the range of the error map is chosen to be $\calM$.

The notion of an error function allows one to define the concept that a system $f^\dag$ on $\calG$ contains a \emph{synchronous model} of  $f$ on $\calM$.

\begin{definition}[Synchronous Model]\label{dfn:synchronous_model}
Consider two systems \eqref{eq:system} defined by $f : \vecL \to \gothX(\calM)$ and $f^\dag : \vecL \to \gothX(\calG)$.
Let $e : \calG \times \calM \to \calM$ be an error function according to Definition~\ref{dfn:error_function}.
The system $f^\dag$ is said to \emph{contain a synchronous model} of $f$ (for the error $e$) if for all initial conditions $\hat{X}_0 \in \calG$ and $\xi_0 \in \calM$ then
\begin{align}
\ddt e (\hat{X}(t; \hat{X}_0),\xi(t;\xi_0)) = 0,
\label{eq:ddt_e_0}
\end{align}
where $\hat{X}(t; \hat{X}_0)$ and $\xi(t;\xi_0)$ are trajectories of the systems $f^\dag$ and $f$, respectively, driven by the same input signal $v(t) \in \vecL$.
\hfill$\Box$\end{definition}

Analogous to how the definition of internal model depends on the reconstruction map $\varphi$, the definition of synchronous model depends on the error function $e : \calG \times \calM \to \calM$ (Def.~\ref{dfn:error_function}). 
It is important to note that Definition \ref{dfn:synchronous_model} is a property of $f^\dag$ (with respect to $f$) rather than a property of $f$. 
This distinction is important in \S\ref{sec:symmetry_construction} when we discuss synchronous lifts (Def.~\ref{dfn:synchronous_lift}). 
If the error function is clear from context then we may omit the reference and state simply that a system $f^\dag$ contains a synchronous model of $f$.

It is straightforward to show that if $f^\dag$ contains a synchronous model of $f$ then it contains an internal model of $f$ for the reconstruction map $\varphi$ defined as in \eqref{eq:reconstruction_map}.
Indeed, fixing $\mr{\xi} \in \calM$, then for any $\hat{X}_0$ define $\xi_0 = \varphi(\hat{X}_0) = e_{\hat{X}_0}^{-1}(\mr{\xi})$.
Since $\ddt e(\hat{X}(t; \hat{X}_0),\xi(t; \xi_0)) \equiv 0$ along the solutions of $f^\dag$ and $f$, it follows that
\[
e_{\hat{X}(t; \hat{X}_0)}(\xi(t; \xi_0)) = e_{\hat{X}_0}(\xi_0) = \mr{\xi}
\]
is constant.
The result follows by inverting $e_{\hat{X}(t; \hat{X}_0)}$. 
That is $\xi(t; \xi_0) = e_{\hat{X}(t; \hat{X}_0)}^{-1}(\mr{\xi}) = \phi(\hat{X}(t; \hat{X}_0))$ \eqref{eq:reconstruction_map}. 
However, a synchronous model is far more than an internal model, and can be thought of (non-rigorously) as describing parallel (in the sense encoded by the error $e$) trajectories to the system trajectories.

\section{Existence of a Synchronous Lift}
\label{sec:symmetry_construction}

The existence of an internal or synchronous model is a property of a pair of systems, $f^\dag : \vecL \to \gothX(\calG)$ and $f : \vecL \to \gothX(\calM)$ (Def.~\ref{dfn:synchronous_model}). 
In this section, we ask the question, `given a system $f : \vecL \to \gothX(\calM)$, does there exist a system $f^\dag : \vecL \to \gothX(\calG)$ that contains a synchronous model?'.  
This is a question about a single system $f$. 
To answer it we introduce the concept of a synchronous lift and then go on to characterise systems $f$ that admit a synchronous lift. 
We begin with a classical definition of a lift. 

\begin{definition}
\label{dfn:lift}
Let  $f : \vecL \to \gothX(\calM)$ be a system on $\calM$. 
Then $f^\dag : \vecL \to \gothX(\calG)$ is a termed a \emph{lift} or \emph{lifted system} of $f$ if there exists a smooth surjective submersion $\varphi : \calG \to \calM$ (reconstruction map) such that 
\[
\tD \varphi f_v^\dag(X) = f_v(\varphi(X))
\]
for all $v \in \vecL$ and $X \in \calG$. 
\hfill$\Box$\end{definition}

Let $X(t)$ be a solution of $\dot{X} = f^\dag_{v(t)}(X(t))$ for some input signal $v(t) : \R \to  \vecL$. 
Since 
\[
\ddt \varphi(X(t)) = \tD \varphi f_v^\dag(X(t)) = f_v (\varphi(X(t))) 
\]
then $\xi(t) = \varphi(X(t))$ is a solution of $f$.
It follows that if $f$ admits a lifted system $f^\dag$ then $f^\dag$ contains an internal model of $f$.  
It is straightforward to verify that $f^\dag = f$ is a lifted system for itself with reconstruction map $\varphi = \id$ the identity map. 
Thus, every system admits a lift (the system itself) and this lift contains an internal model of the system (simply a copy of the system). 
This simple correspondence does not hold for synchronous models. 

\begin{definition}
\label{dfn:synchronous_lift}
Let $f : \vecL \to \gothX(\calM)$ be a system on $\calM$. 
The system $f$ is said to \emph{admit a synchronous lift} if there exists a manifold $\calG$, an error function $e : \calG \times \calM \to \calM$, and a system $f^\dag : \vecL \to \gothX(\calG)$ such that 
\begin{align}
\tD_{\hat{X}} e(\hat{X}, \xi) [f^\dag_v(\hat{X})] = - \tD_{\xi} e(\hat{X}, \xi) [f_v(\xi)],
\label{eq:synchronous_lift}
\end{align}
for all $v \in \vecL$, $\hat{X} \in \calG$ and $\xi \in \calM$.

We also refer to a synchronous lift $g^\dag$ of a vector field $g \in \Lie(\image (g))$ if it satisfies \eqref{eq:synchronous_lift} for all $\xi \in \calM$. 
\hfill$\Box$\end{definition}

Note that if $f^\dag$ is a synchronous lift then 
\begin{align}
\ddt  e (\hat{X}(t),\xi(t)) 
&=
\tD_{\hat{X}} e(\hat{X}, \xi) [f_v^\dag(\hat{X})] + \tD_{\xi} e(\hat{X}, \xi) [f_v(\xi)] = 0,
\label{eq:ddt_e_synch_condition}
\end{align}
and $f^\dag$ contains a synchronous model of $f$. 
However, the two definitions are not identical. 
Definition~\ref{dfn:synchronous_model} is a property of the system $f^\dag$ with respect to a given system $f$. 
In contrast, Definition~\ref{dfn:synchronous_lift} is a property of the system $f$, and may refer to a multitude of possible synchronous lifts $f^\dag$, any one of which contains a synchronous model for $f$.

Note that a synchronous lift $f^\dag$ is a lift in the classical sense (Definition \ref{dfn:lift}) with respect to the reconstruction function $\varphi$ \eqref{eq:reconstruction_map} since
\begin{align}
\tD_{\hat{X}} \varphi(\hat{X})  [f^\dag (\hat{X})] &=
 \tD_{\hat{X}} e_{\hat{X}}^{-1}(\mr{\xi}) [f^\dag(\hat{X})] \notag \\
    &= - (\tD_{\zeta | e_{\hat{X}}^{-1}(\mr{\xi})} e(\hat{X}, \zeta))^{-1} \tD_{X|\hat{X}} e(X, e_{\hat{X}}^{-1}(\mr{\xi})) [f^\dag(\hat{X})] \notag \\
    &= (\tD_{\zeta | e_{\hat{X}}^{-1}(\mr{\xi})} e(\hat{X}, \zeta))^{-1} \tD_{\zeta | e_{\hat{X}}^{-1}(\mr{\xi})} e(\hat{X}, \zeta) [f(e_{\hat{X}}^{-1}(\mr{\xi}))] \notag \\
    &= f(e_{\hat{X}}^{-1}(\mr{\xi})) = f(\varphi(\hat{X})). \label{eq:synchlift_2_lift}
\end{align}
Thus, the question of existence of a system containing a synchronous model is equivalent to asking; `given $f : \calL \to \gothX(\calM)$ does $f$ admit a synchronous lift?' (Def.~\ref{dfn:synchronous_lift}). 
That is, is it possible to construct a manifold $\calG$ and an error function 
$e : \calG \times \calM \to \calM$ and a synchronous lift $f^\dag : \vecL \to \gothX(\calG)$ for $f$?
Note that, a priori the manifold $\calG$ and the error function $e : \calG \times \calM \to \calM$ are not given, let alone the lift  $f^\dag$, and this question is significantly more challenging than asking simply if a given system $f^\dag : \vecL \to \gothX(\calG)$ contains a synchronous model for $f$. 

The following technical lemma relates the Lie algebras generated by a system $f$ and a synchronous lift $f^\dag$ of $f$.
In effect, it states that the Lie algebra generated by $f^\dag$ must contain a synchronous lift for every vector field in the Lie algebra of generated by $f$.

\begin{lemma}
\label{lem:synchronous_vector_fields}
Consider two systems $f^\dag : \vecL \to \gothX(\calG)$ and $f : \vecL \to \gothX(\calM)$ and suppose that $f^\dag$ is a synchronous lift of $f$.
Then, for every vector field $g \in \Lie(\image(f))$, there exists a synchronous lift $g^\dag \in \Lie(\image(f^\dag))$ (Def.~\ref{dfn:synchronous_lift}).
\hfill$\Box$\end{lemma}

\begin{proof}
From Def.~\ref{dfn:synchronous_lift} there exists an error function $e : \calG \times \calM \to \calM$ such that \eqref{eq:synchronous_lift} holds for all $v \in \vecL$. 
Define $L_0(f) = \mathrm{span}\{\image(f)\}$.
For every $k=0,1,...$, recursively define
\begin{align*}
    L_{k+1}(f) = L_k(f) + \set{[g_1,g_2] \in \gothX(\calM)}{g_1,g_2 \in L_k(f)},
\end{align*}
where the sum is the subspace sum of $L_k(f)$ and all brackets of of vector fields in $L_k(f)$.
Through induction over $k$ it will be shown that, for every $g \in L_k(f)$ there exists a synchronous lift $g^\dag \in L_k(f^\dag)$.

For the base case $k=0$, since $f^\dag$ contains a synchronous model of $f$, then $g^\dag := f_v^\dag \in \gothX(\calG)$ is a synchronous lift \eqref{eq:synchronous_lift} of any $g = f_v$.
This holds for $f_0$ and by linearity of \eqref{eq:ddt_e_synch_condition} for each $f_i = f_0 - f_{\eb_i}$ \eqref{eq:affine_system}.
It follows that for any $g \in \spn\{f_0, f_1, \ldots, f_\ell\} = L_0(f)$ there is a synchronous lift $g^\dag \in L_0(f^\dag)$.

Fix $n \in \N$ and assume that for each $g \in L_n(f)$ there exists a synchronous lift $g^\dag \in L_n(f^\dag)$.
Then, for any $g_1, g_2 \in L_n(f)$,
\begin{align*}
    \tD_{\hat{X}} & e(\hat{X}, \xi) [\ad_{g_1^\dag} g_2^\dag](\hat{X})\\
    &= \at{\ddt}{t=0} e(\Flow_{g_2^\dag}^{-\sqrt{t}}\circ\Flow_{g_1^\dag}^{-\sqrt{t}}\circ\Flow_{g_2^\dag}^{\sqrt{t}}\circ\Flow_{g_1^\dag}^{\sqrt{t}} (\hat{X}), \xi) \\
    &= \at{\ddt}{t=0} e(\Flow_{g_1^\dag}^{-\sqrt{t}}\circ\Flow_{g_2^\dag}^{\sqrt{t}}\circ\Flow_{g_1^\dag}^{\sqrt{t}} (\hat{X}), \Flow_{g_2}^{\sqrt{t}}(\xi)) \\
    &\vdots \\
    &= \at{\ddt}{t=0} e(\hat{X}, \Flow_{g_1}^{-\sqrt{t}}\circ\Flow_{g_2}^{-\sqrt{t}}\circ\Flow_{g_1}^{\sqrt{t}}\circ\Flow_{g_2}^{\sqrt{t}}(\xi)) \\
    &= \tD_\xi e(\hat{X}, \xi)[\ad_{g_2} g_1](\xi) \\
    &= - \tD_\xi e(\hat{X}, \xi)[\ad_{g_1} g_2](\xi).
\end{align*}
That is, $\ad_{g_1^\dag}g_2^\dag = [g_1^\dag, g_2^\dag] \in \gothX(\calG)$ and $\ad_{g_1}g_2 = [g_1,g_2] \in \gothX(\calM)$ satisfy \eqref{eq:synchronous_lift}.
By linearity, it follows that there is a synchronous lift $g^\dag \in L_{n+1}(f^\dag)$ for every $g \in L_{n+1}(f)$.

Since the construction of $\Lie(\image(f))$ is countable we conclude by induction that there exists a synchronous lift $g^\dag \in \gothX(\calG)$ for every $g \in \Lie(\image(f))$.
\end{proof}

The following theorem provides a necessary condition for the existence of a synchronous lift: the accessibility Lie algebra of $f:  \vecL \to \gothX(\calM)$ must be finite-dimensional and complete.

\begin{theorem}\label{thm:synchrony_lie_algebra}
Consider a system $f : \vecL \to \gothX(\calM)$.
If $f$ admits a synchronous lift then the Lie algebra generated by the system $\Lie(\image f)$ is complete and finite dimensional. 
\makebox[1mm]{}\hfill$\Box$\end{theorem}

\begin{proof}
If $f$ admits a synchronous lift then there exists a manifold $\calG$, an error $e : \calG \times \calM \to \calM$ and a function $f^\dag : \vecL \to \gothX(\calG)$ that contains a synchronous model of $f$, 
We begin by showing that $\Lie(\image f)$ is finite-dimensional.
Fix an arbitrary element $I \in \calG$ and define the map $\pi : \tT_I \calG \to \gothX(\calM)$ by
\begin{align*}
    \pi_\Delta(\xi) &= \tD_{X|I} e_X^{-1}(e(I, \xi)) [\Delta],
\end{align*}
for all $\Delta \in \tT_I \calG$ and $\xi \in \calM$.

Consider that for any $\xi \in \calM$ and $X \in \calG$ one has $e_X^{-1}(e(X,\xi)) = \xi$.
Then, for any $\Delta \in \tT_I \calG$,
\begin{align*}
    0 & =
    \tD_{X|I} \left(e_X^{-1}(e(X,\xi))\right) [\Delta] \\
    &= \tD_{X|I} e_X^{-1}(e(I,\xi)) [\Delta]
    + \tD_{z|e(I,\xi)} e_I^{-1}(z) \tD_{X|I} e(X,\xi) [\Delta],
\end{align*}
and therefore
\begin{align}\label{eq:pi_map_characterisation}
\pi_\Delta(\xi) &=
- \tD_{z|e(I,\xi)} e_I^{-1}(z) \tD_{X|I} e(X,\xi) [\Delta] \notag\\
&= - \left(\tD e_I(\xi) \right)^{-1} \tD_{X|I} e(X,\xi) [\Delta]
\end{align}

Let $g \in \Lie(\image f)$ be arbitrary.
By Lemma \ref{lem:synchronous_vector_fields}, there exists a synchronous lift $g^\dag \in \gothX(\calG)$ of $g$.
Applying \eqref{eq:pi_map_characterisation},
\begin{align*}
    \pi_{g^\dag(I)}(\xi)
    = \tD_{X|I} e_X^{-1} & (e(I, \xi)) [g^\dag(I)]
    = g(e_I^{-1}(e(I, \xi)))
    = g(\xi),
\end{align*}
for all $\xi \in \calM$.
Since $g \in \Lie(\image f)$ was chosen arbitrarily, we see that $\Lie(\image f)$ lies entirely in the image of $\pi$, which is necessarily finite-dimensional as $\pi$ is a linear map with finite-dimensional domain $\tT_I\calG$.

It remains to show that each vector field $g \in \Lie(\image f)$ is complete.
Let $g \in \Lie(\image f)$ be arbitrary and choose a synchronous lift $g^\dag \in \Lie(\image f^\dag) \subset \gothX(\calG)$ for $g$.
Define $\gamma^\dag : (-\varepsilon, \varepsilon) \to \calG$ to be an integral curve of $g^\dag$ with $\gamma^\dag(0) = I$, and define $\gamma(\xi, t) = e_{\gamma^\dag(t)}^{-1}(e(I,\xi))$ for all $\xi \in \calM$ and $t \in (-\varepsilon, \varepsilon)$.

Recalling  \eqref{eq:synchronous_lift}, one has
\begin{align*}
    \ddt \gamma(\xi, t) &= \ddt e_{\gamma^\dag(t)}^{-1}(e(I,\xi)) \\
    &= \tD_{X | \gamma^\dag(t)} e_X^{-1}(e(I,\xi)) [g^\dag(\gamma^\dag(t))] \\
    &= -(\tD_{\zeta | e_{\gamma^\dag(t)}^{-1}(z)} e(\gamma^\dag(t), \zeta))^{-1}
    \tD_{X|\gamma^\dag(t)} e(X, e_{\gamma^\dag(t)}^{-1}(z))[g^\dag(\gamma^\dag(t))] \\
    &= g(e_{\gamma^\dag(t)}^{-1}(e(I,\xi))) \\
    &= g(\gamma(\xi, t)).
\end{align*}
Hence $\gamma(\xi, t) = \Flow_g^t(\xi)$ for all $\xi \in \calM$ and $t \in (-\varepsilon, \varepsilon)$, and the flow of $g$ exists at every point $\xi$ for a fixed minimum interval $(-\varepsilon, \varepsilon)$.
By the uniform time lemma \cite[Lemma 9.15]{2012_lee_IntroductionSmoothManifolds}, it follows that the flow may be defined for all $t$ and thus $g$ is complete.
Since $g$ was arbitrary, this proves that $\Lie(\image f)$ is complete.
\end{proof}

Although the converse is not true in general, it does hold for controllable systems.
The following lemma provides the constructive process to build a Lie group and group action for a system with finite-dimensional accessibility algebra. 
These will be used later in Theorem \ref{thm:group_action_system} to construct a lifted system $f^\dag$. 

\begin{lemma}\label{lem:faithful_group_action}
    Let $\gothL \subset \gothX(\calM)$ be a Lie subalgebra of vector fields on a manifold $\calM$.
    If $\gothL$ is finite-dimensional and complete, then there exists a unique connected Lie group $\grpG$ along with a faithful action $\phi : \grpG \times \calM \to \calM$ such that $\tD \phi(I) : \gothg \to \gothL$ is a Lie algebra isomorphism.
\end{lemma}

\begin{proof}
Let $\gothg$ is an abstract real Lie algebra that is isomorphic to $\gothL$ by a map $L : \gothg \to \gothL \subset \gothX(\calM)$.
Then let $\grpG^s$ be the unique simply connected Lie group associated with $\gothg$.
By \cite[Theorem XVIII]{1957_palais_GlobalFormulationLie}, there exists a right Lie group action $\phi^s : \grpG^s \times \calM \to \calM$ such that
\begin{align*}
\tD \phi^s_\xi(I) [U] = L(U)(\xi),
\end{align*}
for all $U \in \gothg$ and $\xi \in \calM$.

To obtain a faithful group action, let 
\[
    \grpK = \cset{X \in \grpG^s}{\phi^s(X, \xi) = \xi \quad \forall \xi \in \calM}.
\]
Then $\grpK$ is a discrete normal subgroup of $\grpG^s$, and we may define the quotient Lie group $\grpG = \grpG^s / \grpK$ \cite[Theorem 21.26]{2012_lee_IntroductionSmoothManifolds} along with the quotient action $\phi : \grpG \times \calM \to \calM$.
This action is necessarily faithful, and the Lie algebra of $\grpG$ is unchanged since $\grpK$ is discrete.

For uniqueness of $\grpG$, suppose there exists another connected Lie group $\grpG'$ and faithful action $\phi' : \grpG' \times \calM \to \calM$ such that $\gothg$ is the Lie algebra of $\grpG'$ and $\tD \phi'(I) : \gothg \to \gothL$ is an isomorphism.
Then $\grpG^s$ must also be be the simply connected universal covering of $\grpG'$ and therefore $\grpG' = \grpG^s / \grpK'$ for some discrete normal subgroup $\grpK'$.
Suppose now that $\grpK \neq \grpK'$ and, without loss of generality, assume that there exists some $S \in \grpK$ with $S \not\in \grpK'$.
It follows that the coset $S \grpK' \in \grpG'$ but also that $\phi'(S\grpK',\xi) = \phi(S, \xi) = \xi$, which contradicts the assumption that the action $\phi'$ is faithful.
Therefore, it must be that $\grpK = \grpK'$ and hence $\grpG = \grpG'$, proving that $\grpG$ is indeed unique.
\end{proof}

\begin{theorem}\label{thm:group_action_system}
    Consider a controllable system $f : \vecL \to \gothX(\calM)$ and suppose that the Lie algebra generated by the system $\Lie(\image f)$ is finite-dimensional and complete.
    Then there exists a transitive group action $\phi : \grpG \times \calM \to \calM$ for a Lie group $\grpG$ and a synchronous lift $f^\dag : \vecL \to \gothX(\grpG)$ for $f$ on $\grpG$ (Def.~\ref{dfn:synchronous_lift}).
\hfill$\Box$
\end{theorem}

\begin{proof}
From Lemma \ref{lem:faithful_group_action}, and recalling that $f$ is controllable, we have that there exists a Lie group $\grpG$ and a faithful action $\phi : \grpG \times \calM \to \calM$ such that $L := \tD \phi(I) : \gothg \to \gothL := \Lie(\image f)$ is a Lie algebra isomorphism.
Since the system is controllable, $D_\gothL$ spans the whole tangent space, and the group action $\phi$ must be transitive.

Define a candidate error function $e : \grpG \times \calM \to \calM$ by
\begin{align}
e(\hat{X},\xi) = \phi(\hat{X}^{-1}, \xi).
\label{eq:synch_error}
\end{align}
Smoothness of $e$ is immediate.
The family of partial maps $e_{\hat{X}} = \phi_{\hat{X}^{-1}} = \phi_{\hat{X}}^{-1}$ are diffeomorphisms of $\calM \to \calM$ due to being group actions.
For each $\mr{\xi} \in \calM$, the partial map $\varphi = e^{-1}_{\hat{X}}(\mr{\xi}) = \phi(\hat{X}, \mr{\xi})$ is a surjective submersion due to $\phi$ being transitive.
This verifies that $e$ is an error function according to Definition~\ref{dfn:error_function}.

It remains to show that there is a synchronous lift $f^\dag$ of $f$.
Recall that the group action $\phi$ is obtained by integrating the flow of vector fields in $\gothL$ \cite[Theorem XVIII]{1957_palais_GlobalFormulationLie} by using the isomorphism $L : \gothg \to \gothL$.
Thus, since $\image f \subseteq \gothL$, it is possible to define the affine map
\begin{align*}
    \Lambda : \vecL \to \gothg,
    && \Lambda(v) = L^{-1} (f_v).
\end{align*}
It follows that
\[
\tD \phi_\xi(I) [\Lambda(v)] = L(\Lambda(v)) = f_v(\xi),
\]
for all $v \in \vecL$ and $\xi \in \calM$.
Define the candidate lifted system $f^\dag : \vecL \to \gothX(\grpG)$ to be
\[
f^\dag_v(\hat{X}) := \tD \tL_{\hat{X}} \Lambda(v) = \hat{X} \Lambda(v).
\]
Recalling Definition~\ref{dfn:synchronous_lift} one computes 
\begin{align} 
\tD_{\hat{X}} e(\hat{X}, \xi) [f^\dag_v(\hat{X})] %
= & \tD_{\hat{X}} \phi_\xi (\hat{X}^{-1}) \tD \tL_{\hat{X}} [\Lambda(v)] \label{eq:pf:synch_lift} \\ 
= & \tD_{Z|\Id} \phi_\xi ((\hat{X} Z)^{-1}) [\Lambda(v)] \notag \\ 
= & \tD_{Z|\Id} \phi (\hat{X}^{-1}, \phi_\xi(Z^{-1})) [\Lambda(v)] \notag \\ 
= & \tD \phi_{\hat{X}^{-1}} \cdot \tD_{Z|\Id} \phi_\xi (Z^{-1}) [\Lambda(v)] \notag \\ 
= & \tD \phi_{\hat{X}^{-1}} [- f_v(\xi)]\label{eq:pf_synch} \\ 
= & - \tD_{\xi} e (\hat{X},\xi) [f_v(\xi)] \notag 
\end{align}
where the negative in \eqref{eq:pf_synch} comes from differentiating $Z^{-1}$ at the identity. 
This shows that $f^\dag$ is a synchronous lift and completes the proof.
\end{proof}

We emphasise that Theorem~\ref{thm:group_action_system} is constructive.
If the Lie algebra generated by the system $\Lie(\image f)$ is finite-dimensional and complete, then a synchronous model can be constructed by identifying the Lie group $\grpG$ corresponding to $\Lie(\image f)$ by following the steps in Lemma \ref{lem:faithful_group_action} and then building $f^\dag$ from the associated group action. 
An example of this process is provided in \S\ref{ex:unicycle}.

\section{Fundamental Systems}
\label{sec:fundamental_systems}

Theorems \ref{thm:synchrony_lie_algebra} and \ref{thm:group_action_system} in Section~\ref{sec:symmetry_construction} show that the property that a system admits a synchronous lift is intrinsically linked to existence of a homogeneous structure (transitive Lie group action) on the system state space that is compatible with the system dynamics. 
Systems on homogeneous spaces that are compatible with a group action are termed \emph{equivariant systems} and have been studied in the systems and control literature for many years 
\cite{1985_grizzle_StructureNonlinearControl,2007_bonnabel_ObservateursAsymptotiquesInvariants,2022_mahony_ObserverDesignNonlinear}. 
However, not all equivariant systems admit a synchronous lift and it is relevant to characterise the subclass of systems that do. 
In this section, we define a subset of equivariant systems that we term \emph{fundamental systems} and show that there is a one-to-one correspondence between fundamental systems and systems that admit a synchronous lift. 

Consider a homogeneous manifold $\calM$ with associated Lie group $\grpG$ and transitive right group action $\phi : \grpG \times \calM \to \calM$. 
Every element $u \in \gothg$ induces a \emph{fundamental vector field} $\phi^\sharp_u \in \gothX(\calM)$ \eqref{eq:fundamental_vector_field} on the manifold $\calM$ \cite{1963_kobayashi_FoundationsDifferentialGeometry}. 
Recalling that $\phi$ is a right-handed group action, the mapping $\phi^\sharp_{\cdot} : \gothg \to \gothX(\calM)$ is a Lie algebra homomorphism and is injective if $\phi$ is faithful.
The definition of a \emph{fundamental system} extends the concept of a fundamental vector field to the setting of systems with inputs.

\begin{definition}[Fundamental System]
\label{dfn:fundamental_system}
Let $f: \vecL \to \gothX(\calM)$ be a system on a homogeneous space $\calM$ with transitive group action $\phi : \grpG \times \calM \to \calM$ for a Lie group $\grpG$. 
We say that $f$ is a \emph{fundamental system} if there exists an affine map $\Lambda : \vecL \to \gothg$ for which
\begin{align}\label{eq:fundamental_system}
f_v(\xi) &:= \phi^\sharp_{\Lambda(v)}(\xi) =  \tD \phi_\xi(I)[\Lambda(v)],
\end{align}
for all $\xi \in \calM$ and $v \in \gothg$.
\hfill$\Box$\end{definition}

Fundamental systems extend the existing concepts of group-linear \cite{2010_jouan_EquivalenceControlSystems} and group-affine \cite{2017_barrau_InvariantExtendedKalman} systems already present in the literature. 
The existing concepts are developed for systems on Lie groups and do not capture the key structure of fundamental vector fields that is core to understanding and extending to systems on manifolds. 
This leads to the main result in the paper. 

\begin{theorem}\label{thm:synchrony_fundamental}
Consider a controllable system $f : \vecL \to \gothX(\calM)$ on a smooth manifold $\calM$. 
Then $f$ admits a synchronous lift if and only if $f$ is fundamental. 
\end{theorem}

\begin{proof}
If $f$ is a fundamental system then there is a transitive group action $\phi : \grpG \times \calM \to \calM$ for a Lie group $\grpG$ and an input map $\Lambda : \vecL \to \gothg$ such that $f_v = \phi^\sharp_{\Lambda(u)}$. 
We will term the system 
\begin{align}
\dot{\hat{X}} = f^\dag_v(\hat{X})  := \hat{X} \Lambda(v). 
\label{eq:fundamental_lift}
\end{align}
a \emph{fundamental lift} of $f$. 
Define an error 
\[
e := \phi(\hat{X}^{-1}, \xi).
\]
Then repeating the derivation in \eqref{eq:pf:synch_lift} shows that $f^\dag$ is a synchronous lift for $f$ Def.~\ref{dfn:synchronous_lift}. 

In the reverse direction, assume that $f$ admits a synchronous lift. 
Then by Theorem \ref{thm:synchrony_lie_algebra}, the Lie algebra generated by the system $\Lie(\image f)$ is complete and finite dimensional. 
Theorem \ref{thm:group_action_system} provides a construction for a group $\grpG$, group action $\phi : \grpG \times \calM \to \calM$ and function $\Lambda : \vecL \to \gothg$ such that 
\[
f^\dag_v(\hat{X}) := \tD \tL_{\hat{X}} \Lambda(v) = \hat{X} \Lambda(v)
\]
is a synchronous lift for $f$ with respect to the error $e = \phi(X^{-1}, \xi)$. 
By construction 
\[
f_v(\xi) = \tD \phi_\xi [f^\dag (\Id)] 
= \tD \phi_\xi [\Lambda(v)] = \phi_{\Lambda(v)}^\sharp. 
\] 
It follows that $f$ is a fundamental system.
\end{proof}

Theorems \ref{thm:group_action_system} and \ref{thm:synchrony_fundamental} together show that $\Lie(\image (f))$ must always be a sub-algebra of the algebra $\gothg$ of the symmetry group $\grpG$ of a fundamental system. 
Indeed, for a system $f$, the smallest symmetry group that corresponds to a fundamental structure for that system is the Lie group directly associated with the Lie-algebra $\Lie(\image (f))$ (see Lemma \ref{lem:faithful_group_action}). 
It is often useful, however, to work with a larger symmetry Lie group $\grpG' \geq \grpG$ that has a better algebraic structure or acts on additional degrees of freedom in the system state-space that are not captured directly by the system dynamics. 
For example, Theorem \ref{thm:synchrony_fundamental} assumes that the system $f$ is controllable; that is, that the distribution generated by the accessibility Lie-algebra of the system is full rank. 
If the Lie algebra generated by the system is finite-dimensional and complete but its accessibility distribution $D_{\Lie(\image(f))}$ does not span the whole tangent space, then the Lie group construction provided in Lemma \ref{lem:faithful_group_action} will generate orbits that are integral manifolds immersed in $\calM$ rather than the whole space. 
That is, the group action constructed is not transitive on $\calM$. 
In many cases, it is still possible to construct a larger Lie group $\grpG' \geq \grpG$ that contains the minimal symmetry group $\grpG$ as a subgroup, and an extension of the group action $\left. \phi'\right|_{\grpG} = \phi$ that acts transitively on the whole manifold. 
The synchronous lift in the new degrees of freedom, that act on the inaccessible directions on $\calM$, will be zero to ensure that associated projected vector fields preserve the integral manifold structure of the system. 
Such constructions play a crucial role in allowing observer design for bias, calibration, and other general parameter estimation problems where the associated system model is uncontrollable but still observable. 

\section{Observer Architecture for Fundamental Systems}
\label{sec:observer_Design}

This section explores observer design techniques for fundamental systems.

\begin{figure}[htb]
    \centering
    \includegraphics[width=0.6\columnwidth]{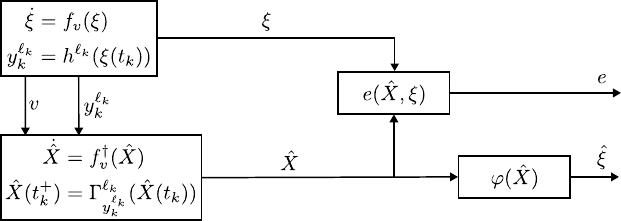}
    \caption{A \emph{Synchronous Observer Architecture}. Inputs $v$ and outputs $y^\ell$, $\ell = 1,..., n$, are used by the synchronous lift $f^\dag$ and correction functions $\Gamma^\ell$, respectively.
    The error between the observer state $\hat{X}$ and the system state $\xi$ is given by the error function $e(\hat{X}, \xi)$.
    The state estimate is obtained through the reconstruction function $\varphi(\hat{X})$.
    }
    \label{fig:observer_architecture}
\end{figure}

\subsection{Synchronous Observer Design for Fundamental Systems}

An observer is an algorithm for estimating the state of a system from its input signal and measurements.
An observer \emph{architecture} is a structural decomposition of the algorithm into components for ease of understanding and implementation. 
The proposed synchronous observer architecture (Figure~\ref{fig:observer_architecture}) is an adaptation for fundamental systems of the classical hybrid observer architecture. 

\begin{definition}\label{dfn:synchronous_architecture}
A \emph{synchronous observer architecture} (Figure~\ref{fig:observer_architecture}) consists of: 
\begin{enumerate}
\item  A fundamental system (Def.~\ref{dfn:fundamental_system}) $f : \vecL \to \gothX(\calM)$ 
\[
f_v = \phi^\sharp_{\Lambda(v)} 
\]
for a Lie group $\grpG$ and associated group action $\phi : \grpG \times \calM \to \calM$, and $\Lambda : \vecL \to \gothg$ affine, along with a collection of measurement functions $h^\ell : \calM \to \calN^\ell$, $\ell = 1,...,n$. 
\item A fundamental lift \eqref{eq:fundamental_lift}  $f^\dag : \vecL \to \gothX(\grpG)$,
\[
f^\dag_v (\hat{X}) = \tD \tL_{\hat{X}} \Lambda(v),
\]
of $f$ along with a collection of \emph{update functions} $\Gamma^\ell : \calN^\ell \to \Diff(\calG)$, 
\[
\Gamma_{y^\ell}^\ell : \grpG \to \grpG. 
\]
\item The synchronous error function $e(\hat{X}, \xi) := \phi(\hat{X}^{-1}, \xi)$ (\eqref{eq:synch_error} and Def.~\ref{dfn:error_function}). 
\item A fixed choice of origin $\mr{\xi} \in \calM$ and the reconstruction function \eqref{eq:reconstruction_map}
\[
\varphi(\hat{X}) = \phi_{\mr{\xi}} (\hat{X}) .
\]
\end{enumerate}
\hfill$\Box$\end{definition}

The \emph{state estimate} generated by a synchronous observer architecture is \begin{align}
    \hat{\xi} = \varphi(\hat{X}) =  \phi(\hat{X}, \mr{\xi}).
    \label{eq:state_reconstruction}
\end{align}
Importantly, the error equals the origin if and only if the state estimate equals the state,
\begin{align*}
    e(\hat{X}, \xi) = \phi(\hat{X}^{-1}, \xi)  = \mr{\xi} && \iff && 
    \hat{\xi} = \phi(\hat{X}, \mr{\xi}) = \xi.
\end{align*}
From synchrony one has that $\dot{e}(\hat{X}(t), \xi(t)) = 0$ in the absence of measurements.
If we define a function $\Lyap : \calM \to \R_+$ with $\Lyap (\mr{\xi}) = 0$ a minimum and appropriate growth properties around $\mr{\xi}$, then $\Lyap (e(\hat{X}, \xi))$ is an effective comparison function to analyse convergence of the observer. 
If there are no measurements then 
\[
\ddt \Lyap(e(\hat{X},\xi)) = \tD \Lyap (e) [\dot{e}] = 0,
\]
and the comparison function is stationary. 
In particular, the only point at which the error can change, and when the comparison function updates, is when the observer updates its state using an update function $\Gamma^\ell$. 
This is a hybrid process that happens instantaneously whenever a new measurement is available. 
The update functions are designed to adjust the observer state to drive the error $e$ to the origin $\mr{\xi}$ as measured by $\Lyap(e) \to 0$. 
That is, for a sequence of measurements $y_1^{\ell_1},y_2^{\ell_2},... \in \calN$ satisfying
\begin{align}
    y_k^{\ell_k} = h^{\ell_k}(\xi(t_k)),
\end{align}
for $t_1 < t_2 < \ldots \in \R$ and $\ell_1, \ell_2, \ldots \in \{1, \ldots, n\}$.
Then instantaneously 
\begin{align}
    \hat{X}(t_k^+) & := \Gamma^{\ell_k}_{h^{\ell_k}(\xi(t_k))} (\hat{X}(t_k)), \label{eq:correction_dynamics} \\
    \hat{\xi}_k^+ & := \varphi(\hat{X}(t_k^+) ) 
\end{align}
for every $k \in \N$. 
In the intervening intervals, the observer state $\hat{X}$ evolves according to the fundamental lift dynamics.

\begin{remark}
    In practice, most observers are implemented digitally and inputs are also received discretely.
    Typically, however, the inputs are received at a far higher rate than the measurements and the proposed architecture is usually appropriate even if the actual fundamental lift is implemented numerically. 
\end{remark}

\subsection{Constructive Observers through Lyapunov Design}
\label{sec:Lyapunov_design}

We propose two theorems that demonstrate some of the key advantages of synchronous observer design.

The first result describes the modularity of synchronous observer designs, extending the particular case previously shown in \cite{2025_vangoor_SynchronousObserverDesign}.
It states that, if a collection of update functions for a synchronous observer design can be shown to each individually decrease a cost function, then combining them will also decrease the cost function.
A particular case of this theorem was exploited in the inertial navigation system (INS) observer design proposed in \cite{2025_vangoor_SynchronousObserverDesign} to individually analyse update functions for each sensor considered, and then add them together to obtain an observer design that seamlessly combines all sensor information.
This property fails even in simple non-synchronous observer designs, such as linear Luenberger observers%
\footnote{Given matrices $A,C_1,L_1,C_2,L_2$ of compatible dimensions, it is possible that $A-L_1C_1$ and $A-L_2C_2$ are Hurwitz but $A-(L_1C_1 + L_2C_2)$ is not. Thus, simply adding correction terms $L_1 C_1$ and $L_2 C_2$ in a Luenberger observer may fail to yield a stable observer}.

\begin{theorem}\label{thm:combined_corrections}
Consider the synchronous observer architecture (Def.~\ref{dfn:synchronous_architecture}). 
Suppose $\Lyap : \calM \to \R^+$ is a cost function and that
\begin{align}\label{eq:individual_lyapunov_correction}
    \Lyap(e(\Gamma^\ell_{h^\ell(\xi)}(\hat{X}), \xi)) - \Lyap(e(\hat{X}, \xi)) \leq 0,
\end{align}
for all states $\xi \in \calM$ and observer states $\hat{X} \in \grpG$, and for each $\ell=1,...,n$.
If the updates \eqref{eq:correction_dynamics} are applied sequentially
then 
\begin{align}
    \ddt \Lyap(e(\hat{X}(t), \xi(t))) &= 0, & t &\in (t_k,t_{k+1}), \notag \\
    \Lyap(e(\hat{X}(t_k^+), \xi(t_k^+))) &\leq \Lyap(e(\hat{X}(t_k), \xi(t_k))), & k &\in \N.
    \label{eq:combined_lyapunov}
\end{align}
and the cost function $\Lyap$ is a Lyapunov-like function \cite[Def.~2.2]{1998_branicky_MultipleLyapunovFunctions} for the resulting error dynamics. 
\hfill$\Box$\end{theorem}

\begin{proof}
Due to synchrony, the dynamics of $e$ between measurement times are constant, i.e. $\dot{e} = 0$ for all $t \in (t_k,t_{k+1})$.
On the other hand, when $t = t_k$, then \eqref{eq:individual_lyapunov_correction} guarantees that the cost function decreases.
Together, these facts mean that the cost function $\Lyap$ is always non-increasing and satisfies the conditions for a hybrid Lyapunov-like function \cite[Def.~2.2]{1998_branicky_MultipleLyapunovFunctions}. 
\end{proof}

The second result we present shows how an update function for a synchronous observer design can be designed from a simple differential correction condition. 
This is important in simplifying the analysis of update terms, since a differential condition is often more easily designed and verified than the full non-linear update term. 
The proposed discretisation method allows one to conduct the analysis of the observer design in continuous-time, but implement the resulting update function as an asynchronous update. 

\begin{theorem}\label{thm:discretised_corrections}
Consider the synchronous observer architecture (Def.~\ref{dfn:synchronous_architecture}).
Suppose $\Lyap : \calM \to \R^+$ is a cost function and $\delta^\ell : \calN^\ell \to \gothX(\grpG)$, $\ell = 1,...,n$, a parametrised family of complete vector fields such that
\begin{align*}
    \tD \Lyap(e(\hat{X}, \xi)) \tD_{\hat{X}} e(\hat{X}, \xi) [\delta^\ell_{h^\ell(\xi)}(\hat{X})] \leq 0,
\end{align*}
for all states $\xi \in \calM$ and observer states $\hat{X} \in \grpG$.
Define the update function $\Gamma^\ell : \calN^\ell \to \Diff(\grpG)$ to be the flow of $\delta^\ell$ for some fixed $\tau_k > 0$,
\begin{align*}
    \Gamma^\ell_{y^\ell}(\hat{X}) &:= \Flow_{\delta^\ell_{y_k^{\ell}}}^{\tau_k} (\hat{X}),
\end{align*}
for all $\hat{X} \in \grpG$ and $y^\ell \in \calN^\ell$.
Then applying \eqref{eq:correction_dynamics} one has 
\[
\Lyap(e(\hat{X}(t_k^+), \xi(t_k^+))) \leq \Lyap(e(\hat{X}(t_k), \xi(t_k))), \quad  k \in \N.
\]
\hfill$\Box$\end{theorem}

\begin{proof}
At the time $t_k$ of the measurement $y_k^{\ell_k} = h^{\ell_k}(\xi(t_k))$, let $\hat{X}_s = \Flow^s_{\delta^{\ell_k}_{h^{\ell_k}(\xi(t))}}(\hat{X}(t))$. Then,
\begin{align*}
    \frac{\td}{\td s} \Lyap(e(\hat{X}_s, \xi(t)))
    &= \tD \Lyap(e(\hat{X}_s, \xi(t_k))) \tD_{\hat{X}_s} e(\hat{X}_s, \xi(t_k)) [\delta^{\ell_k}_{h^{\ell_k}(\xi_{t_k})}(\hat{X}_s)]
    \leq 0,
\end{align*}
for all $s \geq 0$.
Therefore
\begin{align*}
    \Lyap(e(\Gamma^{\ell_k}_{h^{\ell_k}(\xi(t_k))}(\hat{X}(t_k)), \xi(t_k)))
    &= \Lyap(e(\hat{X}_{\tau_k}, \xi(t_k))) \\
    &\leq \Lyap(e(\hat{X}(t_k), \xi(t_k))),
\end{align*}
for any choice of $\tau_k > 0$, as required.
\end{proof}

\section{Examples}
\label{sec:examples}

The following three examples demonstrate the various concepts described in this paper.
The first provides a classic example of a fundamental system and shows how this can be used to construct a synchronous model.
The second shows how the theory developed in Section \ref{sec:symmetry_construction} can be used to identify and construct the Lie group associated with a system when one exists.
The third demonstrates an application of Theorems \ref{thm:combined_corrections} and \ref{thm:discretised_corrections} to extend an existing observer for velocity-aided attitude (VAA) \cite{2023_vangoor_ConstructiveEquivariantObserver}.

The examples make use of the following Lie groups and their Lie algebras.
The special orthogonal group:
\begin{align}\label{eq:SO3_dfn}
    \SO(3) &:= \cset{R \in \R^{3\times 3}}{R^\top R = I_3, \; \det(R) = 1}, \notag \\
    \so(3) &:= \cset{\Omega^\times \in \R^{3\times 3}}{{\Omega^\times}^\top + \Omega^\times = 0_{3\times 3}}.
\end{align}
The special Euclidean group:
\begin{align}\label{eq:SE2_dfn}
    \SE(2) &:= \cset{
        \begin{pmatrix}
            R(\theta) & x \\
            0_{1\times 2} & 1
        \end{pmatrix} \in \R^{3\times 3}
    }{\theta \in [0,2\pi), \; x \in \R^2}, \notag \\
    R(\theta) &:= \begin{pmatrix}
        \cos(\theta) & -\sin(\theta) \\ \sin(\theta) & \cos(\theta)
    \end{pmatrix}, \notag\\
    \se(2) &:= \cset{
        \begin{pmatrix}
            0 & -\omega & u_1 \\ 
            \omega & 0 & u_2 \\ 
            0 & 0 & 0
        \end{pmatrix}
    }{\omega,u_1,u_2 \in \R}.
\end{align}
\subsection{Rotating Bearings}
\label{ex:rotation_bearings}

One particularly well studied example of observer design for a fundamental system is that of bearing estimation on the sphere $\Sph^2$ via an observer design on the group of rotations $\SO(3)$.

Consider a vehicle equipped with a gyroscope providing measurements $\Omega \in \R^3$ of the angular velocity in the body fixed frame.
A vector $\eta$, representing the direction of gravity expressed in the vehicle's body fixed frame, evolves according to
\begin{align}\label{eq:example_sphere_system}
    \dot{\eta} = f_\Omega(\eta) := - \Omega^\times \eta,
\end{align}
where $\Omega^\times \in \so(3)$ is the unique skew symmetric matrix such that $\Omega^\times v = \Omega \times v$ for all vectors $v\in \R^3$.

This system is fundamental.
Consider the transitive right group action $\phi : \SO(3) \times \Sph^2 \to \Sph^2$, given by
\begin{align*}
    \phi(R, \eta) := R^\top \eta,
\end{align*}
and define the map $\Lambda : \R^3 \to \so(3)$ by
\begin{align*}
    \Lambda(\Omega) := \Omega^\times.
\end{align*}
Then
\begin{align*}
    \phi^\sharp_{\Lambda(\Omega)}(\eta)
    &= \tD \phi_\eta(\Id)[\Lambda(\Omega)] \\
    &= \ddso \phi(\exp(s\Lambda(\Omega)), \eta) \\
    &= \ddso \exp(s\Omega^\times)^\top \eta \\
    &= (\Omega^\times)^\top \eta \\
    &= -\Omega^\times \eta \\
    &= f_\Omega(\eta),
\end{align*}
for all $\eta \in \Sph^2$ and $\Omega \in \R^3$.

To obtain a synchronous model, we simply follow the construction proposed in Theorem \ref{thm:group_action_system}.
Define the system on the Lie group $\SO(3)$,
\begin{align*}
    f^\dag : \R^3 \to \gothX(\SO(3)), &&
    f^\dag_\Omega(\hat{R}) := \hat{R} \Lambda(\Omega) = \hat{R} \Omega^\times,
\end{align*}
and define the error function $e(\hat{R}, \eta) := \phi(\hat{R}^\top, \eta) = \hat{R} \eta$.
Then, for the system state $\eta$ and observer state $\hat{R}$, synchrony \eqref{eq:ddt_e_synch_condition} is verified by computing
\begin{align*}
    \ddt e(\hat{R}, \eta)
    &= \ddt \phi(\hat{R}^\top, \eta) \\
    &= \ddt \hat{R}\eta \\
    &= f^\dag_\Omega (\hat{R})\eta + \hat{R}f_\Omega (\eta) \\
    &= (\hat{R} \Omega^\times) \eta + \hat{R} (-\Omega^\times \eta) \\
    &= 0.
\end{align*}
This shows that, indeed, $e$ and $f^\dag$ provide a synchronous observer architecture for the system \eqref{eq:example_sphere_system}.

\subsection{Deriving the Symmetry of a Unicycle}
\label{ex:unicycle}

In this example, we demonstrate how the theory presented in Section \ref{sec:symmetry_construction} may be used to discover a symmetry for a system.
Crucially, we make no use of any prior knowledge of a symmetry of the system in question.

Define the kinematic unicycle as a system $f : \R^2 \to \gothX(\calM)$ on the manifold $\calM = \Sph^1 \times \R^2$,
\begin{align*}
    f_{(\omega, v)}(\theta, x, y) = (\omega, v \cos(\theta), v \sin(\theta)).
\end{align*}
The system is easily seen to be controllable, and this means Theorem \ref{thm:synchrony_fundamental} can be applied.
To check that the Lie algebra generated by the system is finite dimensional, we select a basis for $\image f \subset \gothX(\calM)$ by
\begin{align*}
    f^1(\theta, x, y) &= (1, 0, 0), \\
    f^2(\theta, x, y) &= (0, \cos(\theta), \sin(\theta)).
\end{align*}
Now the Lie algebra generated by the system can be computed recursively (see the proof of Lemma \ref{lem:synchronous_vector_fields}).
The Lie bracket of a vector field with itself is always zero, so only the Lie bracket between $f^1$ and $f^2$ needs to be computed at this stage.
One has
\begin{align*}
    &[f^1, f^2](\theta, x, y)  \\
    &= \tD f^2 (\theta, x, y) [f^1(\theta, x, y)] - \tD f^1 (\theta, x, y) [f^2(\theta, x, y)] \\
    &= (0, - \sin(\theta), \cos(\theta)) - (0, 0, 0) \\
    &= (0, - \sin(\theta), \cos(\theta)).
\end{align*}
This new vector field is linearly independent of $f^1$ and $f^2$ and hence requires a new label,
\begin{align*}
    f^3(\theta, x, y) := (0, - \sin(\theta), \cos(\theta)).
\end{align*}
At the next stage, there are two new Lie brackets that need to be computed.
The bracket between $f^1$ and $f^3$ is given by
\begin{align*}
    &[f^1, f^3](\theta, x, y)  \\
    &= \tD f^3 (\theta, x, y) [f^1(\theta, x, y)] - \tD f^1 (\theta, x, y) [f^3(\theta, x, y)] \\
    &= (0, - \cos(\theta), -\sin(\theta)) - (0,0,0) \\
    &= (0, - \cos(\theta), -\sin(\theta)) \\
    &= - f^2.
\end{align*}
The bracket between $f^2$ and $f^3$ is given by
\begin{align*}
    &[f^2, f^3](\theta, x, y)  \\
    &= \tD f^3 (\theta, x, y) [f^2(\theta, x, y)] - \tD f^2 (\theta, x, y) [f^3(\theta, x, y)] \\
    &= (0,0,0) - (0,0,0) \\
    &= (0,0,0).
\end{align*}
This shows that the Lie algebra generated by the system $f$ is finite-dimensional, since all the vector fields generated are now expressed as linear combinations of other vector fields that were already included in the basis.


The Lie algebra $\gothL = \mathrm{span}\{f^1, f^2, f^3\}$, with 
\[
[f^1, f^2] = f^3, \quad [f^1, f^3] = - f^2, \quad [f^2,f^3] = 0,
\]
is isomorphic (as a Lie algebra) to the special Euclidean Lie algebra $\se(2)$ by
\[
f^1 \triangleq 
\begin{pmatrix} 
0 & -1 & 0 \\ 1 & 0 & 0 \\ 0 & 0 & 0 
\end{pmatrix}
\;
f^2 \triangleq 
\begin{pmatrix} 
0 & 0 & 1 \\ 0 & 0 & 0 \\ 0 & 0 & 0 
\end{pmatrix}
\; 
f^3 \triangleq 
\begin{pmatrix} 
0 & 0 & 0 \\ 0 & 0 & 1 \\ 0 & 0 & 0 
\end{pmatrix}
\]

From this point, the corresponding Lie group may be generated either abstractly, by studying the simply connected group generated by $\se(2)$, or directly, by generating elements of the group by the flows of the vector fields included in the Lie algebra.

The simply connected Lie group generated by $\se(2)$ is, notably, not the familiar Lie group $\SE(2)$.
Instead, it is the universal covering group $\widetilde{\SE}(2)$, whose underlying manifold structure is equivalent to $\R^3$, and whose group product is given by
\begin{align*}
    \begin{pmatrix} t_1 \\ a_1 \\ b_1 \end{pmatrix}
    \cdot \begin{pmatrix} t_2 \\ a_2 \\ b_2 \end{pmatrix}
    &= \begin{pmatrix} 
        t_1 + t_2 \\
        a_1 + \cos(t_1) a_2 - \sin(t_1) b_2 \\
        b_1 + \sin(t_1) a_2 + \cos(t_1) b_2
     \end{pmatrix},
\end{align*}
for all $(t_1,a_1,b_1), (t_2,a_2,b_2) \in \widetilde{\SE}(2)$.
Effectively, the covering group $\widetilde{\SE}(2)$ is the same as $\SE(2)$ except that the `angle' $t$ takes values in $\R$ and can `wind up', in contrast to the angle in $\SE(2)$ which only takes values in $[0, 2\pi)$.

The covering group $\widetilde{\SE}(2)$ acts (right-handedly) on the unicycle state space by integrating the vector fields $f^1,f^2,f^3 \in \gothX(\calM)$ of its Lie algebra,
\begin{align*}
    \phi^s((t,a,b), (\theta, x, y)) :=
    \begin{pmatrix}
        \pi_{\Sph^1}(\theta + t) \\ x + \cos(\theta) a - \sin(\theta) b \\ y + \sin(\theta) a + \cos(\theta) b
    \end{pmatrix}.
\end{align*}
The kernel of $\phi^s$, understood as a group homomorphism $\phi^s : \widetilde{\SE}(2) \to \Diff(\calM)$, is exactly 
\begin{align*}
    \grpK &= \left\{(t,a,b) \in \widetilde{\SE}(2) \;\middle|\; \forall (\theta, x, y) \in \calM,
    \phi^s((t,a,b), (\theta, x, y)) = (\theta, x, y) \right\} \\
    &= \cset{(2n\pi, 0,0) \in \widetilde{\SE}(2)}{n \in \mathbb{Z}}.
\end{align*}
This is isomorphic to the group of integers under addition and is easily verified to be a discrete normal subgroup of $\widetilde{\SE}(2)$.
Finally, the quotient group $\widetilde{\SE}(2) / \grpK$ is seen to be exactly $\SE(2)$, i.e. the familiar symmetry group of the unicycle.
As described in Lemma \ref{lem:faithful_group_action}, $\SE(2)$ is thus the unique Lie group with a faithful action generated by the Lie algebra of the system.

\subsection{Velocity-Aided Attitude}
\label{ex:VAA}

The velocity-aided attitude (VAA) problem we consider here is that of estimating the velocity and orientation of a vehicle provided measurements from an inertial measurement unit (IMU), a GNSS velocity sensor, and a magnetometer.
We refer the reader to \cite{2023_vangoor_ConstructiveEquivariantObserver} for a detailed discussion.
In practical VAA, IMU readings typically arrive at a high rate (at least 500~Hz and often up to 2~kHz) while GNSS and magnetometer readings arrive at far lower rates (around 1~Hz and 5~Hz, respectively).
The IMU can be reasonably treated as a continuous-time signal while the GNSS and magnetometer should be treated as asynchronous measurements, justifying the approach presented in Section \ref{sec:observer_Design}.
In the following example, we will derive the correction terms in continuous-time, and then apply them as asynchronous updates according to Theorems \ref{thm:combined_corrections} and \ref{thm:discretised_corrections}.

The system state consists of the orientation and velocity $(R,v) \in \SO(3) \times \R^3$ of the vehicle, defined with respect to a given inertial frame.
Their dynamics satisfy
\begin{align}\label{eq:vaa_system}
    \ddt (R,v) = f_{(\Omega, a)}(R,v) := (R \Omega^\times, R a + g),
\end{align}
where $\Omega, a, g \in \R^3$ are the measured angular velocity, the measured specific acceleration, and the known gravity vector, respectively.

Define the Lie group $\grpG = (\grpR^3 \times \SO(3)) \ltimes \grpR^3$ with product and inverse given by
\begin{align*}
    (z_1, Q_1, x_1) \cdot (z_2, Q_2, x_2) &:= 
    (z_1 + z_2,\; Q_1 Q_2,\; x_1 + Q_1 x_2 + (I - Q_1) z_2 ), \\
    (z_1, Q_1, x_1)^{-1} &:= (- z_1,\; Q_1^\top,\; - Q_1^\top x_1 - (I - Q_1^\top) z_1).
\end{align*}
This acts on the state space $\calM := \SO(3) \times \R^3$ by
\begin{align*}
    &\phi : \grpG \times \calM \to \calM, \\
    &\phi((z, Q, x), (R,v)) := (R Q,\; v + R x + (I - R) z ).
\end{align*}
Define the lift $\Lambda : \R^3 \times \R^3 \to \gothg$ to be
\[
    \Lambda(\Omega, a) := (g, \Omega^\times, a - g),
\]
and observe that
\begin{align*}
    \phi^\sharp_{\Lambda(\Omega, a)}(R,v) 
    &= \tD \phi_{(R,v)}(I)[\Lambda(\Omega, a)] \\
    &= \tD \phi_{(R,v)}(I)[g, \Omega^\times, a - g] \\
    &= (R \Omega^\times,\; R (a - g) + (I - R) g ) \\
    &= (R \Omega^\times,\; R a + g ) \\
    &= f_{(\Omega, a)}(R,v).
\end{align*}
Thus the system dynamics \eqref{eq:vaa_system} are fundamental under $\phi$.

Denote the state as $\xi = (R,v) \in \calM$ and define the observer state $\hat{X} = (\hat{z}, \hat{Q}, \hat{x}) \in \grpG$ with observer dynamics
\begin{align*}
    \dot{\hat{X}} &= \hat{X} \Lambda(\Omega, a); \\
    \ddt (\hat{z}, \hat{Q}, \hat{x}) 
    &= \ddso (\hat{z}, \hat{Q}, \hat{x}) \cdot (s g, \exp(s \Omega^\times), s(a-g)) \\
    &= \ddso (\hat{z} + s g,\; \hat{Q} \exp(s \Omega^\times),\; 
    \\ &\hspace{3cm}
    \hat{x} + s\hat{Q} (a - g)  + s (I - \hat{Q}) g ) \\
    &=  (g,\; \hat{Q} \Omega^\times,\; \hat{Q} (a - g)  + (I - \hat{Q}) g ) \\
    &=  (g,\; \hat{Q} \Omega^\times,\; \hat{Q} a + g ).
\end{align*}
Define the origin $\mr{\xi} = (I_3, 0) \in \calM$. Then the state estimate is defined by
\begin{align*}
    \hat{\xi} = (\hat{R}, \hat{v}) = \phi(\hat{X}, \mr{\xi})
    = (\hat{Q}, \hat{x}).
\end{align*}
The observer error is given by
\begin{align*}
    e &:= e(\hat{X}, \xi) \\
    &= \phi(\hat{X}^{-1}, \xi) \\
    &= \phi((- \hat{z},\; \hat{Q}^\top,\; - \hat{Q}^\top \hat{x} - (I - \hat{Q}^\top) \hat{z}),\; (R,v)) \\
    &= (R \hat{Q}^\top, \; v + R (- \hat{Q}^\top \hat{x} - (I - \hat{Q}^\top) \hat{z}) + (I-R)(-\hat{z})) \\
    &= (R \hat{Q}^\top, \; v - R \hat{Q}^\top \hat{x} - (I - R\hat{Q}^\top) \hat{z} ).
\end{align*}
It is straightforward to verify that the observer architecture is synchronous by computing
\begin{align*}
    \dot{e}
    &= \ddt (R \hat{Q}^\top, \; v - R \hat{Q}^\top \hat{x} - (I - R\hat{Q}^\top) \hat{z}) \\
    &= (R\Omega^\times \hat{Q}^\top - R\Omega^\times \hat{Q}^\top, \; 
    \\ &\hspace{1.5cm}
    (Ra + g)
    - R \hat{Q}^\top (\hat{Q} a + g)
    - (I - R\hat{Q}^\top) g) \\
    &= (0, \; 
    Ra + g - R a - R \hat{Q}^\top g + R\hat{Q}^\top g - g) \\
    &= (0, \; 0).
\end{align*}

\begin{figure*}[!t]
    \includegraphics[width=\linewidth]{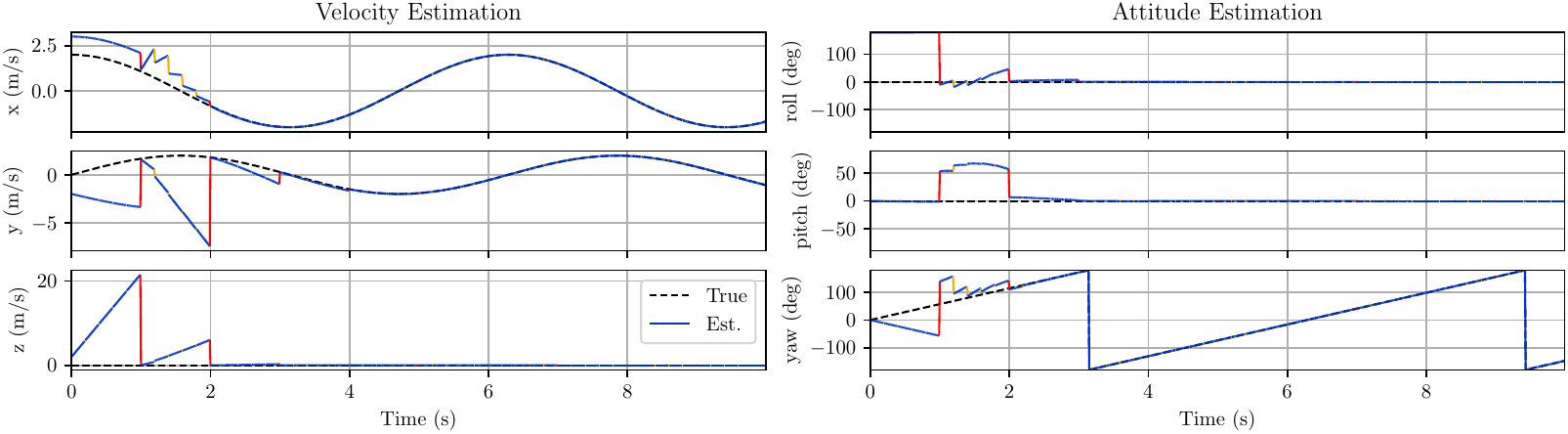}
    \caption{Estimation performance of the velocity-aided attitude observer combining discretised update terms for magnetometer and GNSS velocity measurements.
    The times at which GNSS and magnetometer updates are applied are highlighted in red and orange, respectively.}
    \label{fig:vaa_estimation}
\end{figure*}

\begin{figure}[!t]
    \centering
    \includegraphics[width=0.5\linewidth]{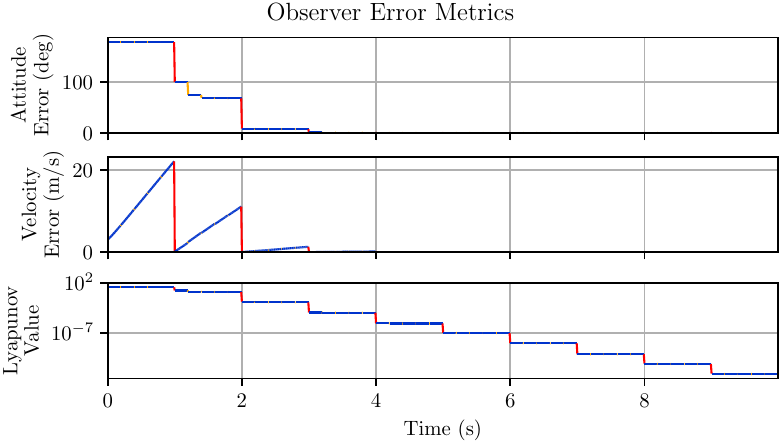}
    \caption{Estimation error and Lyapunov value evolution for the velocity-aided attitude observer combining discretised update terms for magnetometer and GNSS velocity measurements.
    The times at which GNSS and magnetometer updates are applied are highlighted in red and orange, respectively.}
    \label{fig:vaa_errors}
\end{figure}

Define $(R_e, v_e) := e(\hat{X}, \xi)$ and define a candidate Lyapunov function $\Lyap : \calM \to \R$ by
\begin{align}\label{eq:vaa_lyapunov}
    \Lyap(e) := \frac{1}{2} \vert R_e - I_3 \vert^2 + \frac{\alpha}{2} \vert v_e \vert^2,
\end{align}
where $\alpha > 0$ is a constant.
The measurements from the GNSS and magnetometer are given by
\begin{align*}
    y^1 &= h^1(R,v) = v, \\
    y^2 &= h^2(R,v) = R^\top \mr{m},
\end{align*}
respectively, where the $\mr{m} \in \R^3$ is the reference magnetic field direction in the inertial frame (usually $\eb_1 = (1,0,0)$).
We seek to design continuous-time correction terms $\delta^1, \delta^2$ that differentially decrease the candidate Lyapunov function.
We can then discretise these corrections using Theorem \ref{thm:discretised_corrections} to obtain updates and combine them using Theorem \ref{thm:combined_corrections}.

We define the continuous-time correction terms using their right-invariant trivialisations; that is, we write
\begin{align*}
    \delta^\ell_{y^\ell}(\hat{X}) = \Delta^\ell \cdot \hat{X},
\end{align*}
where $\Delta^\ell := \Delta^\ell(\hat{X}, y^\ell) \in \gothg$ depends on the observer state and the measurement.
Under such a correction, the differential change in the error is
\begin{align*}
    \tD_{\hat{X}} e(\hat{X}, \xi) [\delta^\ell_{h^\ell(\xi)}(\hat{X})]
    &= \tD_{\hat{X}} \phi(\hat{X}^{-1}, \xi) [\Delta^\ell \cdot \hat{X}] \\
    &= \tD \phi_{\phi(\hat{X}^{-1}, \xi)}(I) [\Delta^\ell] \\
    &= - \tD \phi_{(R_e, v_e)}(I) [w_{\Delta^\ell}, \Omega_{\Delta^\ell}^\times, u_{\Delta^\ell}] \\
    &= - (R \Omega_{\Delta^\ell}^\times,\; R u_{\Delta^\ell} + (I - R) w_{\Delta^\ell} ).
\end{align*}
Hence, the differential change in the cost function is 
\begin{align}\label{eq:vaa_cost_dynamics}
    \tD \Lyap(e(\hat{X}, \xi)) \tD_{\hat{X}} e(\hat{X}, \xi) [\delta^\ell_{h^\ell(\xi)}(\hat{X})]
    = \langle R_e - I_3,\; R \Omega_{\Delta^\ell}^\times \rangle 
    + \alpha \langle v_e,\; R u_{\Delta^\ell} + (I - R) w_{\Delta^\ell} \rangle.
\end{align}
From \cite{2023_vangoor_ConstructiveEquivariantObserver}, the GNSS velocity correction term
\begin{align}
    &\Delta^1(\hat{X}, y^1) := 
    (
        k_v (y^1-\hat{z}),\;
        k_c ((\hat{x} - \hat{z})^\times (y^1 - \hat{z}))^\times,\;
        k_v (y^1 - \hat{x})
    )\notag
\end{align}
ensures that \eqref{eq:vaa_cost_dynamics} is less than or equal to zero for any constant gains $k_c > 0$ and $ k_v > \frac{k_c}{2 \alpha} > 0$.
From \cite{2012_trumpf_AnalysisNonLinearAttitude}, the magnetometer correction term
\begin{align*}
    \Delta^2(\hat{X}, y^2) :=
    (
        0,\; k_m ( (\hat{Q} y^2) \times \mr{m})^\times,\; 0
    )
\end{align*}
also ensures that \eqref{eq:vaa_cost_dynamics} is less than or equal to zero for any gain $k_m > 0$.

The discretised update terms $\Gamma^\ell_{y_\ell}(\hat{X})$ are applied by solving the following differential equation.
At a time $t_k$ consider the measurement $y^{\ell_k}_k = h^{\ell_k}(\xi(t_k))$.
Then define $\hat{X}^+_{t_k}(\tau)$ as the solution to 
\begin{align*} 
\ddtau \hat{X}^+_{t_k}(\tau) & = \delta^{\ell_k}_{y^{\ell_k}_k}(\hat{X}^+_{t_k}(\tau))
= \Delta^{\ell_k}(\hat{X}^+_{t_k}(\tau), y^{\ell_k}_k) \cdot \hat{X}^+_{t_k}(\tau), \\
\hat{X}^+_{t_k}(0) &= \hat{X}(t_k).
\end{align*}
Then the updated observer state $\hat{X}(t_k^+)$ is obtained by
\begin{align*}
    \hat{X}(t_k^+) 
    &=\Gamma^{\ell_k}_{y^{\ell_k}_k}(\hat{X}(t_k)) 
    := \Flow_{\delta^{\ell_k}_{y^{\ell_k}_k}}^{\tau_k}
    = \hat{X}^+_{t_k}(\tau_k),
\end{align*}
where $\tau_k$ is the update step length at index $k$.

Discretising the corrections can be challenging to do explicitly as it amounts to solving a parametrised nonlinear ODE.
However, the flows of each $\delta^\ell_{y^\ell}$ can also be implemented numerically using any of a variety of numerical integration schemes compatible with manifolds (and particularly with Lie groups), for example \cite{1998_munthe-kaas_RungeKuttaMethodsLie}.
For sufficiently small $\tau_k$, the decrease condition \eqref{eq:individual_lyapunov_correction} can be guaranteed.
Combining the resulting asynchronous updates $\Gamma^1_{y^1}$ and $\Gamma^2_{y^2}$ is as simple as applying them sequentially upon receiving measurements.

Outside of the context of a synchronous observer, it is not obvious how to guarantee the Lyapunov decrease when discretising or combining these correction terms.
However, thanks to Theorems \ref{thm:discretised_corrections} and \ref{thm:combined_corrections}, we know that they may be discretised as flows of vector fields and combined through composition.

To demonstrate the final observer design, we provide simulation\footnote{Code at: \url{https://github.com/pvangoor/synchronous_VAA}} results in Figures \ref{fig:vaa_estimation} and \ref{fig:vaa_errors}.
We considered a vehicle flying in a circle with a forward velocity of 2~m/s and an angular velocity of 1~rad/s.
Precisely, we defined the initial condition of the system to be $(R,v) = (I_3, 2\eb_1)$ and defined the inputs to be $(\Omega, a) = (\eb_3, 2\eb_2 - g)$, with $g = 9.81 \eb_3$.
The observer was initialised with a large initial error, $\hat{X}_0 = (0,\;\exp(0.99\pi\eb_1^\times),\;(3, -2, 2))$.
The results were obtained by integrating the dynamics of the system and observer at 100~Hz using Lie group Euler integration for 10~s.
The GNSS measurements were provided at 1~Hz and the magnetometer measurements were provided at 5~Hz, and the observer gains were set to $k_v = 5.0, k_c = 1.0, k_m = 5.0$.
The flows of $\delta^1_{y^1}$ and $\delta^2_{y^2}$ were obtained by using Lie group Euler integration with 50 integration steps and $\tau_k = 1,0.2$ for GNSS and magnetometer measurements, respectively.

Figure \ref{fig:vaa_estimation} shows the estimates provided by the observer as compared to the true state, and Figure \ref{fig:vaa_errors} shows the estimation errors and Lyapunov value \eqref{eq:vaa_lyapunov} (with $\alpha = 1$) over time.
At the times when measurements are received, the asynchronous updates improve the estimates of attitude and velocity.
The Lyapunov value is constant in between measurements being received, and is instantaneously reduced every time a measurement is available.

\section{Conclusion}

This paper presented a theory of synchronous models and its relationship to fundamental systems and observer design.
First, we defined what it means for one system to contain a synchronous model of another with respect to an error function as an extension of the concept of an internal model.
Second, we showed that a system admits a synchronous lift if and only if has a complete and finite-dimensional Lie algebra.
For controllable systems, we additionally showed that such a system is a fundamental system, meaning it is induced by the fundamental vector fields of a transitive Lie group action, for which we provide a construction directly from the system's Lie algebra.
Third and finally, we showed how the structure of a fundamental systems can be exploited in observer design to discretise and combine update terms while guaranteeing decrease of a cost function.
Synchronous observer designs have already seen success in addressing practical problems such as attitude estimation \cite{2008_mahony_NonlinearComplementaryFilters,2012_trumpf_AnalysisNonLinearAttitude}, velocity-aided attitude \cite{2023_vangoor_ConstructiveEquivariantObserver}, inertial navigation systems \cite{2025_vangoor_SynchronousObserverDesign}, and homography estimation \cite{2011_hamel_HomographyEstimationSpecial}, and this paper provides a framework for future theoretical development and practical applications.


\bibliographystyle{abbrv}
\bibliography{2025_Synchrony}

\end{document}

%% file: preamble_arXiv.tex
\usepackage{hyperref}
\usepackage{microtype}

\usepackage{amssymb}
\usepackage{amsmath}
\usepackage{amsthm}
\usepackage{mathdots}
\usepackage{mathtools}

\usepackage{tensor}

\usepackage{urwchancal}
\DeclareFontFamily{OT1}{pzc}{}
\DeclareFontShape{OT1}{pzc}{m}{it}{<-> s * [1.10] pzcmi7t}{}
\DeclareMathAlphabet{\mathpzc}{OT1}{pzc}{m}{it}

\usepackage{graphicx}
\usepackage{ifpdf}
\ifpdf
\usepackage{epstopdf}
\PrependGraphicsExtensions{.svg}
\fi

\usepackage{xypic}

\newtheorem{theorem}{Theorem}[section]
\newtheorem{lemma}[theorem]{Lemma}

\newtheorem{definition}[theorem]{Definition}

\newtheorem{remark}[theorem]{Remark}



%% file: notation_defs.tex
\providecommand{\N}{\mathbb{N}}
\providecommand{\R}{\mathbb{R}}

\providecommand{\SO}{\mathbf{SO}}

\providecommand{\SE}{\mathbf{SE}}

\providecommand{\grpG}{\mathbf{G}}

\providecommand{\grpK}{\mathbf{K}}

\providecommand{\grpR}{\mathbf{R}}

\providecommand{\gothg}{\mathfrak{g}}

\providecommand{\gothL}{\mathfrak{L}}

\providecommand{\gothX}{\mathfrak{X}} 

\providecommand{\so}{\mathfrak{so}}
\providecommand{\se}{\mathfrak{se}}

\providecommand{\Sph}{\mathrm{S}}

\providecommand{\calG}{\mathcal{G}}

\providecommand{\calL}{\mathcal{L}}
\providecommand{\calM}{\mathcal{M}}
\providecommand{\calN}{\mathcal{N}}

\providecommand{\vecL}{\mathbb{L}}

\providecommand{\cset}[2]{\left\{ {#1} \;\middle|\; {#2} \right\}}

\providecommand{\tT}{\mathrm{T}} 

\providecommand{\eb}{\mathbf{e}} 

\providecommand{\Id}{I} 

\providecommand{\tL}{\mathrm{L}} 
\providecommand{\tR}{\mathrm{R}} 

\DeclareMathOperator{\spn}{span}

\DeclareMathOperator{\ad}{ad}
 
\DeclareMathOperator{\image}{im}

\providecommand{\id}{\mathrm{id}} 

\providecommand{\Lyap}{\mathcal{L}} 

\DeclareFontEncoding{LS1}{}{}
\DeclareFontSubstitution{LS1}{stix}{m}{n}
\DeclareSymbolFont{stixletters}{LS1}{stix}{m}{it}
\DeclareMathAccent{\cev}{\mathord}{stixletters}{"91}
\DeclareMathAccent{\vec}{\mathord}{stixletters}{"92}
\DeclareMathAccent{\vecev}{\mathord}{stixletters}{"95}

\providecommand{\td}{\mathrm{d}}
\providecommand{\tD}{\mathrm{D}}
\providecommand{\ddt}{\frac{\td}{\td t}}

\providecommand{\ddtau}{\frac{\td}{\td \tau}}

\providecommand{\mr}[1]{\mathring{#1}} 

\usepackage{accents}
\usepackage{mathtools}
\makeatletter
\providecommand{\scirc}{%
    \hbox{\fontfamily{\rmdefault}\fontsize{0.4\dimexpr(\f@size pt)}{0}\selectfont{\raisebox{-0.52ex}[0ex][-0.52ex]{$\circ$}}}}

\makeatother

\makeatletter
\providecommand{\ucirc}{%
    \hbox{\fontfamily{\rmdefault}\fontsize{0.4\dimexpr(\f@size pt)}{0}\selectfont{\raisebox{0.0ex}[0ex][-0.52ex]{$\circ$}}}}

\makeatother

\mathchardef\mhyphen="2D

\providecommand{\idx}[5][]{
\ifthenelse{\isempty{#1}}
{\tensor*[_{#4}^{#3}]{#2}{_{#5}}}
{\tensor*[_{#4}^{#3}]{#2}{^{#1}_{#5}}}
}

